\documentclass[11pt,a4]{article}

\usepackage{amsmath}
\usepackage{latexsym,amsmath}
\usepackage{amsthm}
\usepackage{amssymb}
\usepackage{enumerate}

\usepackage{epsfig, verbatim}
\usepackage{fullpage}
\usepackage{graphicx}
\usepackage{color}

\newtheorem{theorem}{Theorem}
\newtheorem{lemma}[theorem]{Lemma}

\newtheorem{conjecture}[theorem]{Conjecture}

\newtheorem{claim}{Claim}

\newtheorem{remark}{Remark}

\title{Recoloring bounded treewidth graphs}
\author{Marthe Bonamy \and Nicolas Bousquet}

\begin{document}

\maketitle









\begin{abstract}
Let $k$ be an integer. Two vertex $k$-colorings of a graph are \emph{adjacent} if they differ on exactly one vertex. A graph is \emph{$k$-mixing} if any proper $k$-coloring can be transformed into any other through a sequence of adjacent proper $k$-colorings. Any graph is $(tw+2)$-mixing, where $tw$ is the treewidth of the graph (Cereceda 2006). We prove  that the shortest sequence between any two $(tw+2)$-colorings is at most quadratic, a problem left open in Bonamy et al. (2012).

Jerrum proved that any graph is $k$-mixing if $k$ is at least the maximum degree plus two. We improve Jerrum's bound using the grundy number, which is the worst number of colors in a greedy coloring.

\end{abstract}


\section{Introduction}

Reconfiguration problems (see~\cite{Gopalan09,ItoD11,ItoD09} for instance) consist in finding step-by-step transformations between two feasible solutions such that all intermediate results are also feasible. Such problems model dynamic situations where a given solution is in place and has to be modified, but no property disruption can be afforded. In this paper our reference problem is vertex coloring. 

In the whole paper, $G=(V,E)$ is a graph where $n$ denotes the size of $V$ and $k$ is an integer. For standard definitions and notations on graphs, we refer the reader to~\cite{Diestel}.
A \emph{(proper) $k$-coloring} of $G$ is a function $f : V(G) \rightarrow \{ 1,\ldots,k \}$ such that, for every edge $xy$, $f(x)\neq f(y)$. 

Two $k$-colorings are \emph{adjacent} if they differ on exactly one vertex. The \emph{$k$-recoloring graph of $G$}, denoted $R_k(G)$, is the graph whose vertices are $k$-colorings of $G$, with the adjacency defined above. Note that two colorings equivalent up to color permutation correspond to distinct vertices.
The graph $G$ is \emph{$k$-mixing} if $R_k(G)$ is connected. Cereceda, van den Heuvel and Johnson characterized the $3$-mixing graphs and provided an algorithm to recognize them~\cite{Cereceda09,CerecedaHJ11}.

Determining if a graph is $k$-mixing is $\mathbf{PSPACE}$-complete for $k \geq 4$~\cite{BonsmaC07}. The \emph{$k$-recoloring diameter} of a $k$-mixing graph is the diameter of $R_k(G)$. In other words, it is the minimum $D$ for which any $k$-coloring can be transformed into any other through a sequence of at most $D$ adjacent $k$-colorings.
The \emph{mixing number} of $G$ is the minimum integer $m(G)$ for which $G$ is $k$-mixing for every $k \geq m(G)$. It can be arbitrarily larger than the minimum $k$ for which $G$ is $k$-mixing~\cite{Cereceda}. Indeed, for complete bipartite graphs minus a matching, the chromatic number equals two and the mixing number is arbitrarily large (see Fig.~\ref{fig:completmoinsmatching}).

\begin{figure}
\centering
\includegraphics[scale=0.7]{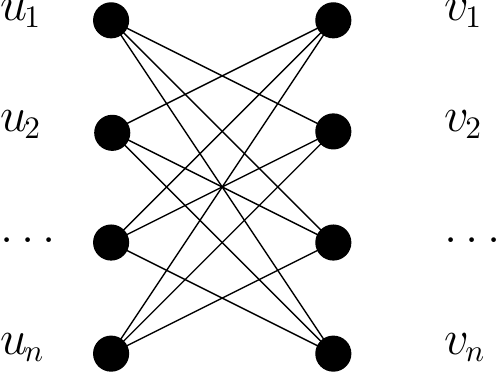}
\caption{In the $n$-coloring where $u_i,v_i$ are given the same color, no vertex can be recolored.}
\label{fig:completmoinsmatching}
\end{figure}


Jerrum~\cite{Jerrum95} proved that $m(G) \leq \Delta(G)+2$, where $\Delta(G)$ denotes the maximum degree. 
Let $x_1,\ldots,x_n$ be an order $\mathcal{O}$ on $V$. We denote by $N(v)$ the neighborhood of $x$. In the \emph{greedy coloring} $C(G,\mathcal{O})$ of $G$ relative to $\mathcal{O}$, every $x_i$ has the smallest color that does not appear in $N(x_i) \cap \{ x_1,\ldots,x_{i-1} \}$. Introduced in~\cite{Christen79}, the \emph{grundy number} $\chi_g(G)$ is the maximum, over all the orders $\mathcal{O}$, of the number of colors used in $C(G,\mathcal{O})$. So $\chi_g(G)$ is the worst number of colors in a greedy coloring of $G$.

\begin{theorem}\label{thm:grundy}
For any graph $G$, if $k \geq \chi_g(G)+1$, then $G$ is $k$-mixing and the $k$-recoloring diameter is at most $4 \cdot \chi_g(G) \cdot n$.
\end{theorem}

Section~\ref{sect:grundy} is devoted to a proof of Theorem~\ref{thm:grundy}. Theorem~\ref{thm:grundy} improves Jerrum's bound since $\chi_g(G) \leq \Delta(G)+1$. And it can be arbitrarily smaller, on stars for instance. 

Besides, the bound is tight on some graphs satisfying $m(G)=\Delta(G)+2=\chi_g(G)+1$~\cite{Cereceda} (see Fig.~\ref{fig:completmoinsmatching}, complete bipartite graphs minus a matching for instance). Nevertheless, $m(G)$ is not bounded by a function of $\chi_g(G)$ since for any $k$, some tree $T_k$ satisfies $\chi_g(T_k)=k$ and $m(T_k)=3$~\cite{Beyer82}.
In addition, unlike the maximum degree, the grundy number is NP-hard to compute~\cite{Zaker06}.

Graphs of treewidth $k$, being $k$-degenerate, are $(k+2)$-mixing~\cite{Cereceda}. However, the best known upper-bound on the recoloring diameter is exponential. In Section~\ref{sect:prooftw}, we prove that the recoloring diameter is polynomial for bounded treewidth graphs. Given a graph $G$ and an integer $k$, it is NP-complete to decide if $tw(G) \leq k$~\cite{ArnborgCP87}. Nevertheless, for every fixed $k$, there is a linear time algorithm to decide if the treewidth is at most $k$ (and find a tree decomposition)~\cite{Bodlaender93}.

\begin{theorem}\label{thm:tw}
For every graph $G$, if $k \geq tw(G)+2$, then $G$ is $k$-mixing and its $k$-recoloring diameter is at most $2 \cdot (n^2+n)$.
\end{theorem}

The quadratic bound on the recoloring diameter was known for chordal graphs~\cite{BonamyJ12}, but its generalization to bounded treewidth graphs was left open. As shown in the case of chordal graphs~\cite{BonamyJ12} (which is a subclass of graphs of treewidth $\omega(G)$), the mixing number is tight, and the recoloring diameter is tight up to a constant factor.

\section{Preliminaries}
Let us first recall some classical definitions on sets. Let $X$ and $Y$ be two subsets of $V$. The set $X \setminus Y$ is the subset of elements $x \in X$ such that $x \notin Y$. By abuse of notation, given a set $X$ and an element $x$, $X\setminus x$ denotes $X \setminus \{ x\}$. The \emph{size} $|X|$ of $X$ is its number of elements.

Let $G=(V,E)$ be a graph. The \emph{neighborhood} of a vertex $x$, denoted by $N(x)$ is the subset of vertices $y$ such that $xy \in E$. The length of a path is its number of edges. The \emph{distance} between two vertices $x$ and $y$, denoted $d(x,y)$, is the minimum length of a path between these two vertices. When there is no path, the distance is infinite. The \emph{distance} between two $k$-colorings of $G$ is implicitely the distance between them in the recoloring graph $R_k(G)$. Let us first recall a classical result on recoloring.

\begin{lemma}\label{lemma:clique}
If $k \geq n+1$, any $k$-coloring of $K_n$ can be transformed into any other by recoloring every vertex at most twice.
\end{lemma}
\begin{proof}
Let $\alpha, \beta$ be two colorings of $K_n$. Let $D$ be the digraph on $n$ vertices with an arc $xy$ if $\beta(x)=\alpha(y)$. Informally $xy$ is an arc if the color of $y$ (in $\alpha$) prevent the recoloring of $x$. No vertices of $\beta$ are colored identically, so for every $x$, $d^+(x) \leq 1$. By symmetry on $\alpha$, $d^-(x) \leq 1$. Hence $D$ is a union of directed paths and of circuits. 

Let $x_0,x_1,\ldots,x_k,x_0$ be a circuit. Since $k \geq n+1$, $x_0$ can be recolored with a free color. We have $d^+(x_k)=0$. And the number of circuit strictly decreases. Indeed $x_0,\ldots,x_k$ is still an oriented path and $d^+(x_k)=0$. Since every vertex has an outdegree at most $1$, no vertex in $\{x_0,\ldots,x_k\}$ can be on a circuit, i.e. the number of circuit strictly decreases.
So by recoloring every vertex at most once,  we can assume that there is no circuit. Therefore a vertex $x$ satisfies $d^+(x)=0$. We can recolor $x$ with $\beta(x)$ which does not create any arc in $D$.
\end{proof}

\section{Mixing number and grundy number}\label{sect:grundy}

This section is devoted to a proof of Theorem~\ref{thm:grundy}. An \emph{optimal coloring} of $G$ is a greedy $\chi(G)$-coloring.
Theorem~\ref{thm:grundy} is derived from the following lemma.

\begin{lemma}\label{lem:grundystep}
Let $G$ be a graph on $n$ vertices, and $k \geq \chi_g(G)+1$. For any $k$-coloring $\alpha$ of $G$ and any optimal coloring $\beta$ of $G$, we have $d(\alpha,\beta)\leq 2 \cdot \chi(G) \cdot n$.
\end{lemma}

\begin{proof}
Let us prove it by induction on $\chi(G)$.

If $\chi(G)=1$, $G$ has no edge. Thus we can recolor the vertices independently. In $n$ steps, we can transform $\alpha$ into $\beta$.

Assume now that $\chi(G)\geq 2$. For any integer $i$ and any coloring $\alpha$, $V_i^\alpha$ is the set of vertices of color $i$ in $\eta$. Iteratively on $i$ from $1$ to $\ell$, we recolor the vertices of $V^{\alpha}_i$ with the smallest color for which the coloring is still proper. The resulting coloring $\gamma$ of $G$ is the greedy coloring relative to the order $V^{\alpha}_1, V^{\alpha}_2, \ldots, V^{\alpha}_{\ell}$. Hence $\gamma$ is an (at most) $\chi_g(G)$-coloring. In addition, $d(\alpha, \gamma) \leq n$, since no vertex is recolored twice. 

Since no vertex is colored with $\chi_g(G)+1$ in $\gamma$ and $k \geq \chi_g(G)+1$, we recolor vertices of $V_1^{\gamma}\setminus V_1^{\beta}$ with color $\chi_g(G)+1$. We then recolor vertices of $V_1^{\beta}$ with $1$ if needed. The resulting coloring $\delta$ satisfies $V_1^{\delta}=V_1^{\beta}$. In addition, $d(\gamma,\delta) \leq n$, for no vertex is recolored twice. 

Let us now prove that the induction hypothesis holds on $G'=G(V \setminus V_1^{\beta})$ with $\ell-1$ colors. We have $\chi(G') = \chi(G)-1$. In addition, $\chi_g(G') < \chi_g(G)$. Indeed, assume that there is an order $\mathcal{O}$ on $V \setminus V_1^{\beta}$ such that $\chi_g(G')= \chi_g(G)$. Consider the order $\mathcal{O}'=(V_1^{\beta},\mathcal{O})$ on $V$. Every vertex of $\mathcal{O}$ has a neighbor on $V_1^{\beta}$ (since $\beta$ is optimal), then the greedy coloring relative to $\mathcal{O}'$ needs $\chi_g(G)+1$ colors for $G$ which is impossible. So we can apply the induction hypothesis on $G'$ with $k-1$ colors (the color $1$ is forgotten). This ensures that $G'$ can be recolored in $2 \cdot (\chi(G)-1) \cdot |V(G')| \leq 2 \cdot (\chi(G)-1) \cdot n$ steps.

Consequently, $d(\alpha,\beta)\leq d(\alpha,\gamma)+d(\gamma,\delta)+d(\delta,\beta) \leq 2 \cdot \chi(G) \cdot n$.
\end{proof}

\section{Bounded treewidth graphs}\label{sect:prooftw}

This aim of this section consists in proving Theorem~\ref{thm:tw}. A \emph{tree} is a connected graph without cycles. In order to avoid confusion, its vertices are called \emph{nodes}. A \emph{tree decomposition} of $G$ is a tree $T$ such that:
\begin{itemize}
 \item To every node $u$ of $T$, we associate a \emph{bag} $B_u \subseteq V$.
 \item For every edge $xy$ of $G$, there is a node $u$ of $T$ such that both $x$ and $y$ are in $B_u$.
 \item For every vertex $x \in V$, the set of nodes of $T$ whose bags contain $x$ form a non-empty subtree in $T$.
\end{itemize}
The \emph{size} of a tree decomposition $T$ is the largest number of vertices in a bag of $T$, minus one. The \emph{treewidth} $tw(G)$ of $G$ is the minimum size of a tree decomposition of $G$.

A \emph{chordal graph} is a graph that admits a perfect elimination ordering: that is, the vertices of the graphs can be ordered $v_1,v_2,\cdots,v_p$ in such a way that the neighborhood of any vertex $v_i$ in $\{v_1,v_2,\vdots,v_{i-1}\}$ forms a clique. Any chordal graph $G$ admits a tree decomposition whose bags are the maximal cliques of $G$.

Actually, the tree decomposition of any graph $G$ can be viewed as a chordal graph $H$ with vertex set $V(G)$ that admits $G$ as a subgraph ($H$ is a \emph{surgraph} of $G$). 
Informally, we transform step-by-step any $(tw+2)$-coloring of a graph into a $(tw+2)$-coloring of a ''good'' chordal surgraph with the same treewidth.

We first introduce particular tree decompositions, called complete tree decompositions. In such decompositions, all the bags have exactly the same size and any two adjacent bags differ on exactly one vertex. Two vertices are parents if their subtrees are, in some sense, adjacent. A $V$-coherent coloring is a coloring where parents are colored identically.

The proof is divided into two parts. First we prove that the distance between $V$-coherent colorings is linear. We then prove that any coloring can be transformed into a $V$-coherent coloring with a quadratic number of recoloring steps as long as the number of colors is at least $tw(G)+2$.


\subsection{Families}

\begin{figure}
 \centering
  \includegraphics[scale=1]{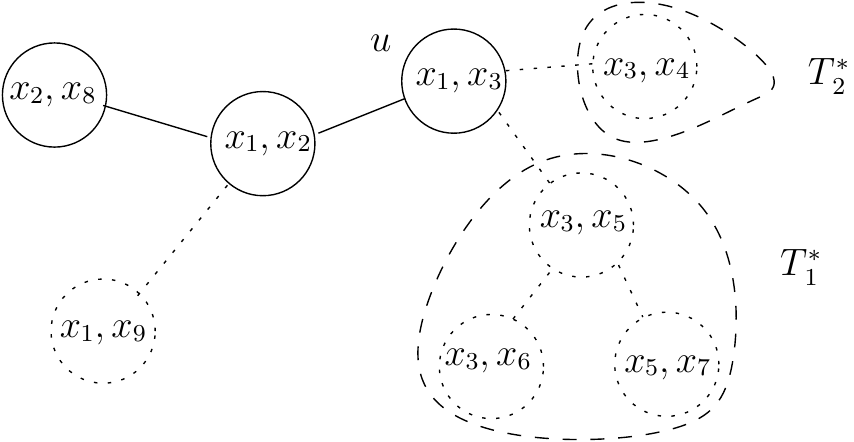}
  \caption{A $1$-complete tree decomposition $T$.}
  \label{fig:completedecomp}
\end{figure}

A tree decomposition $T$ of a graph $G$ is \textit{$\ell$-complete} when every bag has size $\ell+1$ and any two adjacent nodes $u,v$ satisfy $B_u \cap B_v = \ell$. In other words, for every edge $uv$ of $T$, there exists a vertex $x$ such that $x = B_u \setminus B_v$. Let $X \subseteq V$. The tree decomposition $T[V \setminus X]$ is the same tree as $T$ except that the bag of every node $u$ is $B_u \setminus X$, and that every edge $uv$ of $T$ is contracted if $B_u\setminus X \subseteq B_v \setminus X$. In Fig.~\ref{fig:completedecomp}, the full-line edges subtree is $T[V \setminus \{ x_4,x_5,x_6,x_7,x_9 \}]$. In this subsection, we recall classical properties of complete tree decompositions. The first remark is an immediate consequence of the definition.

\begin{remark}\label{rem:subtreelcomplete}
Any connected subtree of an $\ell$-complete tree decomposition is still $\ell$-complete.
\end{remark}

A \emph{baby} is a vertex of $V$ that appears in exactly one bag $B_u$, where $u$ is a leaf of $T$. Note that all the neighbors of a baby $x$ are in $B_u$. In Fig.~\ref{fig:completedecomp}, vertex $x_8$ is a baby.

\begin{remark}\label{rem:lcomplete}
Let $T$ be an $\ell$-complete tree decomposition. If $x$ is a baby then $T[V \setminus x]$ is $\ell$-complete.
\end{remark}
\begin{proof}
Let $u$ be the unique node whose bag contains $x$. Then the only modified bag in $T[V \setminus x]$ is $B_u$. Let $v$ be the father of $u$ in $T$. Since $T$ is complete, $B_u \setminus B_v = x$ in $T$, so the edge $uv$ is contracted in $T[V \setminus x]$. Therefore $T[V \setminus x]$ is exactly $T \setminus u$ which is $\ell$-complete by Remark~\ref{rem:subtreelcomplete}.
\end{proof}

We first prove that every graph admits complete tree decompositions. Then we derive from it the notion of parents and family between vertices of $G$.

\begin{lemma}\label{completetree}
For every graph $G$, if $n-1 \geq \ell \geq tw(G)$ then $G$ admits an $\ell$-complete tree decomposition.
\end{lemma}

\begin{proof}
A tree decomposition of $G$ is \emph{minimal} when every bag has size at most $tw(G)+1$ and no bag is contained in another. Every graph admits a minimal tree decomposition. Indeed, by definition of $tw(G)$, every graph $G$ admits a tree decomposition $T$ whose every bag has size at most $tw(G)+1$. And if two adjacent nodes $u,v$ in $T$ verify $B_u \subseteq B_v$, then the edge $uv$ can be contracted.

Let $T$ be a minimal tree decomposition of $G$.
We build inductively an $\ell$-complete tree decomposition $T_c$ of $G$ such that every bag of $T$ is contained in a bag of $T_c$.

If $n=\ell+1$, then the tree decomposition consisting of a single node with bag $V(G)$ is $\ell$-complete.

If $n\geq \ell+2$, then $T$ has at least two nodes since every vertex is contained in at least one bag. Let $u$ be a leaf of $T$ and $v$ be the neighbor of $u$. Since $T$ is minimal, there is a vertex $x$ in $B_u \setminus B_v$. Note that $x$ is a baby. Otherwise the subset of nodes whose bags contain $x$ would not be a subtree of $T$ since $x\notin B_v$ and $u$ is a leaf of $T$. Let $T'=T[V \setminus x]$.

By induction hypothesis, $G \setminus x$ admits an $\ell$-complete tree decomposition $T'_c$ where every bag of $T'$ is contained in a bag of $T'_c$. So some node $w$ of $T'_c$ satisfies $(B_u \setminus x) \subseteq B_w'$. Since $|B'_w|= \ell+1 \geq tw(G)+1$, some vertex $y$ of $B_w'$ is not in $B_u$. We consider $T_c$ built from $T'_c$ by adding a leaf $u'$ attached on $w$ whose bag is $(B_w' \cup  x ) \setminus y$. Then $T_c$ is an $\ell$-complete tree decomposition of $G$ with the required property with regards to $T$.
\end{proof}

Let $T$ be a complete tree decomposition. Note that $|B_u \setminus B_v|=|B_v \setminus B_u|=1$ for every edge $uv$. Two vertices $x,y \in V$ are \textit{$T$-parents} if there are two adjacent nodes $u,v$ of $T$, such that $x=B_u \setminus B_v$, and $y=B_v \setminus B_u$. In other words, vertices $x$ and $y$ are $T$-parents if the subtree of the nodes containing $x$ in their bags and the subtree of the nodes containing $y$ in their bags do not intersect, but are connected by an edge ($uv$ in this case). Also remark that the notion of parents is symmetric: if $x$ is a parent of $y$ then $y$ is a parent of $x$.

The \emph{family relation} is the transitive closure of the parent relation. A \emph{family} is a class of the family relation. In Fig.~\ref{fig:completedecomp}, the families are $\{x_1,x_4,x_5,x_6,x_8\}$ and $\{x_2,x_3,x_7,x_9\}$. The partition induced by the families is called the \emph{family partition}.  In Fig.~\ref{fig:completedecomp}, vertices $x_2$ and $x_3$ are parents.

\begin{remark}\label{rem:family}
The family partition of any $\ell$-complete tree decomposition exists and is unique. Each family contains exactly one vertex in every bag. So there are $\ell+1$ families, which are stable sets.
\end{remark}

\begin{proof}
By induction on $T$. If $T$ has a single node $u$, then no vertex has a parent. So each family is a single vertex. 

Assume $T$ has at least two nodes. Let $u$ be a leaf of $T$ and $v$ be its adjacent node. Note that the family partitions of $T$ are the extensions of those of $T \setminus u$. The vertices $x=B_u \setminus B_v$ and $y=B_v \setminus B_u$ are parents and $y$ is the unique parent of $x$.
Since $u$ is a leaf of $T$, $T\setminus u$ is still $\ell$-complete by Remark~\ref{rem:subtreelcomplete}. 

By induction, $B_v$ contains exactly one vertex of every family of the unique family partition of $T[V \setminus u]$. Since $B_u =B_v \cup x \setminus y$, and since $y$ is the unique parent of $x$, we can uniquely extend the partition by adding $x$ in the family of $y$. Besides, in $B_u$ there is exactly one vertex of each family.
\end{proof}

\subsection{Coherent colorings}

Let $T$ be an $\ell$-complete tree decomposition of $G$. A coloring $\alpha$ is \emph{$X$-coherent} (relatively to $T$) if for every $x,y \in X$ which are parents, $\alpha(x)=\alpha(y)$ and for every bag $B$ and every $x \in X$, if $x \in B$, then, in $B$, only $x$ is colored with $\alpha(x)$.
Note that since parents are non-adjacent in the graph by Remark~\ref{rem:family}, coherent colorings can be proper. Note also that $V(G)$-coherent colorings are $\ell$-proper coloring.

The subsection is organized as follows. First we define the notion of merged graphs. Then we prove that distance between $V(G)$-coherent colorings is linear. And we finally provide some recoloring lemmas concerning $(V \setminus B_u)$-coherent colorings. All these tools will be used in the next subsection.

Let $G$ be a graph and $\mathcal{C}$ be a stable set. The \emph{merged graph} on $\mathcal{C}$ is the graph $G$ where vertices of $\mathcal{C}$ are identified into a vertex $z$ and $xz$ is an edge if there exists a vertex $y \in \mathcal{C}$ such that $xy$ is an edge. A coloring $\gamma$ of the merged graph can be \emph{extended} on the whole graph by coloring every vertex of $\mathcal{C}$ with $\gamma(z)$. For any stable sets $\mathcal{C}_1,\mathcal{C}_2,\cdots,\mathcal{C}_p$ with $\mathcal{C}_i \cap \mathcal{C}_j = \emptyset$ for any $i \neq j$, the \emph{merged graph} on $\mathcal{C}_1,\mathcal{C}_2,\cdots,\mathcal{C}_p$ is the graph obtained from $G$ by merging successively $\mathcal{C}_1, \mathcal{C}_2 \cdots \mathcal{C}_p$.

\begin{remark}\label{rem:merged}
Let $\mathcal{C}$ be a stable set. Let $\alpha', \beta'$ be two colorings of the merged graph on $\mathcal{C}$ and $\alpha,\beta$ be their extended colorings. If $\alpha'$ can be transformed into $\beta'$ by recoloring each vertex at most $t$ times, then $\alpha$ can be transformed into $\beta$ by recoloring every vertex at most $t$ times.
\end{remark}
\begin{proof}
We just have to follow the recoloring process of $\alpha'$ into $\beta'$. If the recolored vertex is not the merged vertex, then do the same recoloring for the extended colorings. Otherwise, we recolor the vertices one after another in the extended graph. All these are proper since $\mathcal{C}$ is a stable set.
\end{proof}

\begin{lemma}\label{compatibleson}
Let $k \geq tw(G)+2$. If every $k$-coloring of $G$ can be transformed into $V$-coherent coloring with at most $f(n)$ recolorings, then the $k$-recoloring diameter of $G$ is at most $2\cdot(f(n)+n)$.
\end{lemma}

\begin{proof}
Let $\alpha, \beta$ be two $k$-colorings of $G$. By assumption, there are two $V$-coherent colorings $\gamma_\alpha$ and $\gamma_\beta$ such that $d(\alpha,\gamma_\alpha) \leq f(n)$ and $d(\beta,\gamma_\beta) \leq f(n)$. 

Let us prove that $d(\gamma_\alpha,\gamma_\beta)\leq 2n$. By definition, all the vertices of a same family are colored identically in $\gamma_\alpha$. The same holds for $\gamma_\beta$. Let $G'$ be the merged graph where every family is identified into a same vertex. By Remark~\ref{rem:family}, the family partition is unique, so both $\gamma_\alpha$ and $\gamma_\beta$ are extensions of $\gamma_\alpha'$ and $\gamma_\beta'$ colorings of $G'$. Every pair of vertices of $G'$ have distinct colors in $\gamma_\alpha$ (and in $\gamma_\beta$). So $G'$ can be considered as a clique on $tw(G)+1$ vertices (since there are $tw(G)+1$ families).

Therefore Lemma~\ref{lemma:clique} and Remark~\ref{rem:merged} ensures that $d(\gamma_\alpha,\gamma_\beta)\leq 2n$.
Since $d(\alpha,\beta) \leq d(\alpha,\gamma_\alpha)+d(\gamma_\alpha,\gamma_\beta)+d(\gamma_\beta,\beta)$, Lemma~\ref{compatibleson} holds.
\end{proof}

Let us first make some observation for the two forthcoming lemmas. Let $T$ be a tree and $u$ be a node of $T$. We can consider that $T$ is rooted on $u$. Then $w$ is a \emph{father of $v$} if $vw$ is an edge and $v$ is not in the connected component of $u$ in $T \setminus w$ (and $T_u=T$). The \emph{tree rooted on $v$}, denoted by $T_v$, is the connected component of $v$ in $T \setminus w$. Let us first prove some stability on $(V \setminus B_u)$-coherent colorings.

\begin{lemma}\label{claim:everywhere}
Let $T$ be an $\ell$-complete tree decomposition and $u,v$ be two nodes of $T$. Let $\alpha$ be a $(V \setminus B_u)$-coherent coloring where color $a$ does not appear in $B_u$. 

If a vertex of $B_v$ is colored with $a$, every bag of $T_v$ contains a vertex colored with $a$.
\end{lemma}
\begin{proof}
Assume by contradiction that a node $w$ of $T_v$ does not contain a vertex of colored with $a$ in its bag. Choose $w$ in such a way $w$ is as near as possible from $v$ in $T$. Then the father $w'$ of $w$ contains a vertex $y$ of colored with $a$. 

The vertex $y$ is not in $B_u$ since $y$ is colored with $a$. Let $z= B_w \setminus B_{w'}$. We have $z \notin B_u$ since $z \notin B_{w'}$ and $w'$ is the father of $w$. Since $\alpha$ is $(V\setminus B_u)$-coherent, we have $\alpha(y)=\alpha(z)$. But $\alpha(y)=a$, a contradiction.
\end{proof}

\begin{lemma}\label{pfclaim1}
Let $k,\ell$ be two integers with $k \geq \ell+2$.
Let $T$ be an $\ell$-complete tree decomposition and $u$ be a node of $T$. Let $\alpha$ be a $(V \setminus B_u)$-coherent $k$-coloring where color $a$ does not appear in $B_u$. 

Then by recoloring every vertex of $V \setminus B_u$ at most once, we can obtain a  $(V \setminus B_u)$-coherent $k$-coloring where no vertex is colored with $a$.
\end{lemma}
\begin{proof}
Let us prove it by induction on $\ell$. We enforce the induction hypothesis with the following: if a vertex $x$ is recolored, then there is a vertex $z$ in the family of $x$ such that $\alpha(x)=\alpha(z)$. 

If $\ell=0$, then the graph has no edge. Let $x$ be the vertex of $B_u$. Color $a$ can be eliminated by recoloring every vertex at most once. If $y$ is recolored, then $\alpha(y)=a$. Since $\alpha(x) \neq a$ and there is a unique family, the enforced hypothesis holds.

Otherwise, by Claim~\ref{claim:everywhere}, if the color a vertex of $B_v$ is colored with $a$, the color $a$ appears in every bag of $T_v$. Choose $v$ in such a way the father of $v$ does not contain $a$ in its bag. Let $b$ be a color which does not appear in $B_v$. 

Consider the graph $G'$ induced by the vertices of $T_v$ where vertices colored with $a$ in $\alpha$ and forget the color $a$ from the color set. $T_v$ is an $(\ell-1)$-complete tree decomposition of $G'$. Both $\ell$ and $k$ decrease by one. Thus, by induction hypothesis, the color $b$ can be eliminated from $T_v$ by recoloring every vertex of $T_v \setminus B_v$ at most once. And the obtained coloring $\beta$ is $(V \setminus B_v)$-coherent.

\begin{claim}
 If two parents are not colored the same, then one of them is in $B_u$.
\end{claim}
\begin{proof}
Assume by contradiction that two parents $x,y \notin B_u$ satisfies $\beta(x) \neq \beta(y)$. Since $\alpha$ is $(V \setminus B_u)$-coherent, $\alpha(x)=\alpha(y)$. Since only vertices of $V(T_v) \setminus B_v$ are recolored, $x$ or $y$, say w.l.o.g. $x$, are in $V(T_v) \setminus B_v$. So $y$ is in $V_{T_v}$. Indeed otherwise $x$ and $y$ cannot be parents.

Since $\beta$ is $(V(T_v) \setminus B_v)$-coherent, $y \in B_v$. So $x$ is recolored during the process. Therefore, there is a vertex of the family of $x$ in $V(T_v) \setminus B_v$ which is not colored as $y$ in $V(T_v)$. Since parents in $\alpha$ are colored identically except if one of them is in $B_u$ means that $y \in B_u$, a contradiction.
\end{proof}

We can repeat this operation on the other rooted subtrees maximum by inclusion which contains color $a$. No vertex is recolored twice since the subtrees are independent. Indeed, otherwise it means that the father of a root of a subtree contains a vertex colored with $a$.
\end{proof}

\subsection{Obtaining a $V$-coherent coloring}
In order to prove Theorem~\ref{thm:tw}, Lemma~\ref{compatibleson} ensures that we just have to transform any coloring into a $V$-coherent coloring in $n^2$ recolorings.
For any subtree $T'$ of $T$, $B_{T'}$ denotes $\cup_{v \in T'} B_v$.

\begin{lemma}\label{lem:mainlemma}
Let $T$ be a $tw(G)$-complete tree decomposition. For every $\ell$-coloring $\alpha$ of $G$, there is a $V$-coherent coloring $\gamma_\alpha$ such that $d(\alpha,\gamma_{\alpha}) \leq n^2$.
\end{lemma}
\begin{proof}
The proof consists in a recoloring algorithm. We treat vertices one after the other, considering vertices that have at most one parent not yet treated. In other words, we treat babies of the remaining tree-decomposition. Our invariant will ensure that, when $X$ is treated, the current coloring is $X$-coherent. When a new vertex $x$ is treated, we just have to transform the current coloring in order to obtain a $(X \cup \{ x \})$-coherent coloring. At the end of the procedure, the whole vertex set is treated, and then the current coloring is $V$-coherent.

Let us now describe more formally the invariants. The set $F_i$ represents treated vertices at step $i$. Initially, no vertex is treated, so $F_0=\emptyset$. The coloring $c_i$ is the current $\ell$-coloring at the end of step $i$. Initially the coloring is $\alpha$, so $c_0=\alpha$. The invariants at the end of step $i$ are:
\begin{enumerate}[(i)]
\item\label{p1} $F_{i-1} \subset F_i \subseteq V$, and $|F_i|=i$.
\item\label{p2} $T[V \setminus {F_i}]$ is a $\min(tw(G),|V\setminus F_i|)$-complete tree decomposition of $G\setminus F_i$.
\item\label{p3} $c_i$ is an $\ell$-coloring of $G$ obtained from $c_{i-1}$ by recoloring vertices of $F_i$ at most twice.
\item\label{p4} $c_i$ is $F_i$-coherent.
\end{enumerate}

We proceed iteratively on $i$ from $1$ to $n$. Let $u$ be a leaf of $T[V \setminus F_i]$ and $x$ be a baby contained in $B_u$. We want to add $x$ in $F_i$. Denote by $F_{i+1}$ the set $F_i \cup x$. By Remark~\ref{rem:lcomplete} and since $x$ is a baby, $T[V \setminus F_{i+1}]$ is a complete tree decomposition. Thus (\ref{p1}) and (\ref{p2}) are immediately verified. The following consists in proving (\ref{p3}) and (\ref{p4}).


A \emph{residual component} is a connected component of $T \setminus T[V \setminus F_i]$. Informally, a residual component is a subtree of the tree decomposition containing already treated vertex. A \emph{residual component of $u$} is a residual component containing a node adjacent to $u$. Note that vertices which appear in a bag of such a residual component are included in $F_i \cup B_u$.
In Fig.~\ref{fig:completedecomp}, subtrees $T_1^*$ and $T_2^*$ are the residual components of $u$ in $T[V \setminus \{ x_4,x_5,x_6,x_7,x_9 \}]$.

Let $F$ be the union of the residual components on $u$. And let $T^*$ be the subtree $\{u\} \cup F$. Let us consider the graph $G'$ restricted to the vertices of $T^*$. Let $a$ be a color which does not appear in $B_u$. Note that the coloring $c_i$ restricted to $G'$ is $(V(G') \setminus B_u)$-coherent. Indeed, the vertices of $V(G') \setminus B_u$ are in $F_i$, and the coloring $\alpha$ is $F_i$-coherent. 

Therefore, Lemma~\ref{pfclaim1} can be apply. So $c_i$ can be transformed into a $(B_{T^*}\setminus B_u)$-coherent coloring of $G'$ where no vertex is colored with $a$. Every vertex of $F_i$ is recolored at most once. Note that since vertices of $B_u$ are not recolored, the obtained coloring is proper on the whole graph.

Recall that no vertex of $B_u$ are in $F_i$. Therefore, if two vertices of $F_i$ are parents, either they are both in $B_{T^*}$ or both are not in $B_{T^*}$. So the resulting coloring is $F_i$-coherent.

Thus all the members of the family of $x$ which are in $F_i$ can be recolored with $a$, as the vertex $x$ itself. Every vertex is recolored at most once. 
Finally every vertex is recolored at most twice. So the resulting coloring $c_{i+1}$ satisfies condition (\ref{p3}) and (\ref{p4}).

This operation is repeated until $F_i=V$, that is, $i=|V|$. When the last vertex is treated the coloring is $V$-coherent by (\ref{p4}). It follows from (\ref{p3}) that to recolor $G$ from $\alpha$ to $\gamma_\alpha=c_{n}$, it suffices to recolor each vertex $x$ at most $2\cdot (n-i+1)$, where $i$ is the smallest such that $x \in F_i$. Thus, on the whole, it suffices to make $2\cdot \frac{n(n+1)}{2}=n^2+2n$ recolorings.

The analysis can be slightly improved. Indeed, the vertex $x_i$ treated at step $i$ is recolored at most once (since vertices of $B_u$ are not recolored in the first part of the recoloring algorithm). Therefore, every vertex is recolored at most $1+2 \cdot(n-i)$ times, which finally ensures that $d(\alpha,\gamma_\alpha) \leq n^2$.
\end{proof}

\section{Further work}

Graphs of treewidth at most $k$ are $k$-degenerate graphs. The $(k+2)$-recoloring diameter of $k$-degenerate graphs at most $2^n$~\cite{Cereceda}. Note that the bound on the number of colors is optimal since $K_n$ is $(n-1)$-degenerate. Does the class of $k$-degenerate graphs have a polynomial $(k+2)$-recoloring diameter? Or, a weaker question, can we obtain a polynomial recoloring diameter when the number of color increases?

This question seems very challenging. The class of $k$-degenerate graphs also contains some sub-classes that are themselves interesting. One of the most famous is the class of planar graphs (which are $5$-degenerate).
\begin{conjecture}
For any planar graph $G$ and any integer $k$, if $k \geq 7$ then $R_k(G)$ has a polynomial diameter.
\end{conjecture}
This bound of $7$ would be optimal since there are planar graphs that are not $5$-mixing (see Fig.~\ref{fig:mix5}) or not $6$-mixing (see Fig.~\ref{fig:mix6}).

\begin{figure}
\centering
\parbox[b]{3in}{
\centering
 \includegraphics[scale=0.6]{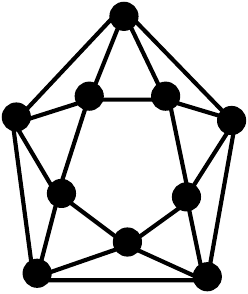}
 \caption{A planar graph that is not 5-mixing.}
 \label{fig:mix5}
}
\qquad
\parbox[b]{3in}{
\centering
 \includegraphics[scale=0.6]{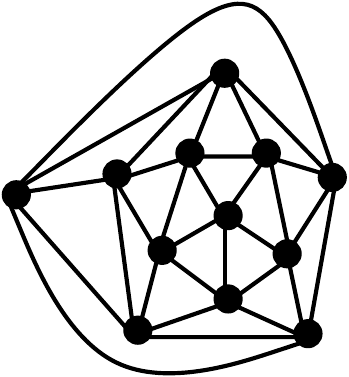}
 \caption{A planar graph that is not 6-mixing.}
 \label{fig:mix6}
 }
\end{figure}

Note that outerplanars graphs have a quadratic recoloring diameter since they have treewidth at most $2$. The quadratic lower bound is optimal~\cite{BonamyJ12} (see Fig.~\ref{fig:outerplanar}).

\begin{figure}
\centering
 \includegraphics[scale=0.9]{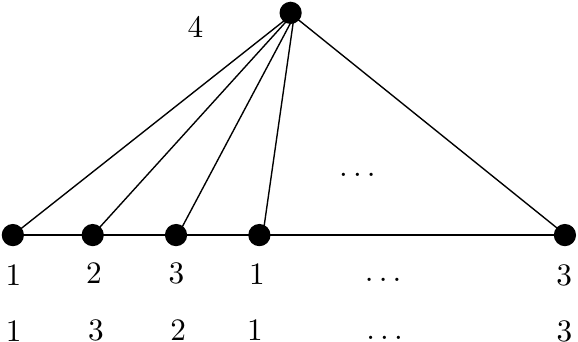}
 \caption{A $3$-colorable outerplanar graph which has a quadratic recoloring diameter.}
 \label{fig:outerplanar}
\end{figure}

Let us consider, as in \cite{Cereceda}, the complete bipartite graph on $2n$ vertices minus a matching (see Fig.~\ref{fig:completmoinsmatching}). The mixing number equals $n+1$, and the chromatic number equals two. Consequently, the mixing number of this family of graphs cannot be bounded by a function of its chromatic number. The same holds for any graph class containing all bipartite graphs. In particular, the mixing number of comparability graphs, perfectly orderable graphs, and perfect graphs cannot be bounded by a function of the chromatic number. This answers an open question of \cite{BonamyJ12} for perfect graphs. 

Another interesting point is the existence of a hamiltonian cycle in the recoloring graph. In other words, is it possible to find a sequence of distinct recolorings which contains all the propers colorings and such that the consecutive colorings are adjacent? Consider for instance $2$-colorings of stable sets on $n$ vertices. The corresponding graph is the $n$-dimensionnal hypercube. Such a graph admits a hamiltonian cycle, known as Gray code. Gray codes, and their generalization, were extensively studied (see~\cite{Savage96} for a survey).

\bibliographystyle{plain}
\bibliography{biblitw}

\begin{thebibliography}{10}

\bibitem{ArnborgCP87}
S.~Arnborg, D.~Corneil, and A.~Proskurowski.
\newblock Complexity of finding embeddings in a k-tree.
\newblock {\em SIAM Journal on Algebraic Discrete Methods}, 8(2):277--284,
  1987.

\bibitem{Beyer82}
T.~Beyer, S.~M. Hedetniemi, and S.~T. Hedetniemi.
\newblock A linear algorithm for the grundy number of a tree.
\newblock In {\em Proceedings of the thirteenth southeastern conferenceon
  combinatorics, graph theory and computing}, 1982.

\bibitem{Bodlaender93}
Hans~L Bodlaender.
\newblock A linear time algorithm for finding tree-decompositions of small
  treewidth.
\newblock {\em Proceedings of the twentyfifth annual ACM symposium on Theory of
  computing}, 25(6):226--234, 1993.

\bibitem{BonamyJ12}
M.~Bonamy, M.~Johnson, I.~Lignos, V.~Patel, and D.~Paulusma.
\newblock Reconfiguration graphs for vertex colourings of chordal and chordal
  bipartite graphs.
\newblock {\em Journal of Combinatorial Optimization}, pages 1--12, 2012.

\bibitem{BonsmaC07}
P.~Bonsma and L.~Cereceda.
\newblock Finding paths between graph colourings: {PSPACE}-completeness and
  superpolynomial distances.
\newblock In {\em MFCS}, volume 4708 of {\em Lecture Notes in Computer
  Science}, pages 738--749, 2007.

\bibitem{Cereceda}
L.~Cereceda.
\newblock {\em Mixing Graph Colourings}.
\newblock PhD thesis, London School of Economics and Political Science, 2007.

\bibitem{Cereceda09}
L.~Cereceda, J.~van~den Heuvel, and M.~Johnson.
\newblock Mixing 3-colourings in bipartite graphs.
\newblock {\em Eur. J. Comb.}, 30(7):1593--1606, 2009.

\bibitem{CerecedaHJ11}
L.~Cereceda, J.~van~den Heuvel, and M.~Johnson.
\newblock Finding paths between 3-colorings.
\newblock {\em Journal of Graph Theory}, 67(1):69--82, 2011.

\bibitem{Christen79}
C.~Christen and S.~Selkow.
\newblock Some perfect coloring properties of graphs.
\newblock {\em Journal of Combinatorial Theory, Series B}, 27(1):49 -- 59,
  1979.

\bibitem{Diestel}
R.~Diestel.
\newblock {\em Graph Theory}, volume 173 of {\em Graduate Texts in
  Mathematics}.
\newblock Springer-Verlag, Heidelberg, third edition, 2005.

\bibitem{Gopalan09}
P.~Gopalan, P.~Kolaitis, E.~Maneva, and C.~Papadimitriou.
\newblock The connectivity of boolean satisfiability: Computational and
  structural dichotomies.
\newblock {\em SIAM J. Comput.}, pages 2330--2355, 2009.

\bibitem{ItoD11}
T.~Ito, E.~Demaine, N.~Harvey, C.~Papadimitriou, M.~Sideri, R.~Uehara, and
  Y.~Uno.
\newblock On the complexity of reconfiguration problems.
\newblock {\em Theor. Comput. Sci.}, 412(12-14):1054--1065, 2011.

\bibitem{ItoD09}
T.~Ito, M.~Kamiński, and E.~Demaine.
\newblock Reconfiguration of list edge-colorings in a graph.
\newblock In {\em Alg. \& Data Struct.}, volume 5664 of {\em Lecture Notes in
  Computer Science}, pages 375--386. 2009.

\bibitem{Jerrum95}
M.~Jerrum.
\newblock A very simple algorithm for estimating the number of k-colorings of a
  low-degree graph.
\newblock {\em Random Structures \& Algorithms}, 7(2):157--165, 1995.

\bibitem{Savage96}
Carla Savage.
\newblock A survey of combinatorial gray codes.
\newblock {\em SIAM Review}, 39:605--629, 1996.

\bibitem{Zaker06}
M.~Zaker.
\newblock Results on the grundy chromatic number of graphs.
\newblock {\em Discrete Mathematics}, 306(23):3166 -- 3173, 2006.

\end{thebibliography}

\end{document}